\documentclass[a4paper]{article}

\usepackage{graphicx,color}

\usepackage{amsmath,amsfonts,amssymb,amsthm}
\newtheorem{defi}{Definition}
\newtheorem{teo}{Theorem}
\newtheorem{prop}[teo]{Proposition}

\newtheorem{rmk}{Remark}

\newcommand{\Sc}{\mathrm{Sc}}

\begin{document}
	
	\date{}
	
	\title{Modified Laplace-Beltrami quantization of natural Hamiltonian systems with quadratic constants of motion}
	\author{Claudia Maria Chanu, 
		Luca Degiovanni, Giovanni Rastelli \\  Dipartimento di Matematica, \\ Universit\`a di Torino.  Torino, via Carlo Alberto 10, Italia.\\ \\ e-mail: claudiamaria.chanu@unito.it \\ luca.degiovanni@gmail.com \\ giovanni.rastelli@unito.it }
	
	\maketitle

	\begin{abstract} 
		It is natural to investigate if the  quantization of an integrable or superintegrable classical Hamiltonian systems is still integrable or superintegrable. We study here this problem in the case of  natural Hamiltonians with constants of motion quadratic in the momenta. The procedure of quantization here considered, transforms the Hamiltonian  into the  Laplace-Beltrami operator plus a scalar potential. In order to transform the constants of motion into symmetry operators of the quantum Hamiltonian, additional scalar potentials, known as quantum corrections, must be introduced, depending on the Riemannian structure of the manifold. We give here a complete geometric characterization of the quantum corrections necessary for the case considered. St\"ackel systems are studied in particular details. Examples in conformally and non-conformally flat manifolds are given.
	\end{abstract}

\section{Introduction}
When a classical natural Hamiltonian $H$
defined on the cotangent bundle of a Riemannian
manifold $M$  admits a quadratic in the momenta first integral $K$, 
the construction of pairwise commuting second order
differential operators associated with $H$ and $K$ 
is a non-trivial problem. 
One of the possible quantization rules is the
so-called Laplace-Beltrami quantization (or minimal quantization,  or Carter quantization):
if we consider the cotangent bundle $T^*M$ of a $N$-dimensional Riemannian manifold $(M, \mathbf g)$, then, 
the Laplace-Beltrami (LB) quantization associates with each quadratic in the momenta function (without  terms linear in the momenta)  $$K=\frac 12 K^{ab}p_ap_b+W$$  
defined by  the symmetric two-tensor $\mathbf K$ and the scalar $W$, the second-order differential operator on functions on $M$ 
\begin{equation*}
\widehat{K}=-\frac {\hbar^2}2 \Delta_K +W=-\frac {\hbar^2}2 \nabla_a K^{ab} \nabla_b +W,
\label{qschr2}
\end{equation*}
where $\nabla_a$ is the covariant derivative w.r.t. the Levi-Civita connection of the metric $\mathbf g$.
In particular, the operator associated with a geodesic Hamiltonian is (up to the constant factor $-\hbar^2/2$) the  Laplace-Beltrami operator $\Delta$. 

When the configuration manifold $M$ is conformally flat, other ways of matching quadratic first integrals of natural Hamiltonians (including $H$ itself) with differential operators on a Hilbert space have been employed, in order to obtain integrability or superintegrability of the quantum versions of specific Hamiltonian systems, see for example the articles \cite{D3,He}. 
Essentially, in both papers the quantization procedure requires a ``quantum   correction'' of the Laplacian: an additional scalar term proportional to the scalar curvature $\Sc$ appears in the  quantum Hamiltonian and in its symmetry operators.
The quantum corrections are necessary in order to allow the commutation with the Hamiltonian operator of one or more of the quantized constants of motion.

Namely, in \cite{He} (beside other two coordinate dependent techniques of quantization) it is shown that the maximal superintegrability of a classical $N$-dimensional Hamiltonian system is kept after quantization only if to the natural Hamiltonian $$H=\frac 12 g^{ab}p_ap_b+V$$ is associated the symmetry operator $$\widehat H_c=-\frac {\hbar^2}2 \left(\Delta-\frac{N-2}{4(N-1)}\Sc\right)+V,$$ where $\Delta$ is the Laplace-Beltrami operator and  $\Sc$ is the scalar curvature of the metric $\mathbf g$.  
The operator $\Delta_c$
\begin{equation}\label{dc}
\Delta_c=\Delta-\frac{N-2}{4(N-1)}\Sc,
\end{equation}
 is the well known conformally invariant Laplacian (also known as Yamabe operator).
A specific study on symmetries and conformal symmetries of the conformally invariant Laplacian is given in 
\cite{Rad}, where the quantization procedure considered is the so-called ``conformally - equivariant quantization" developed in \cite{D1},\cite{D2}. In this type of quantization, equivariance  w.r.t. the action of the conformal symmetry group is required and the classical configuration manifolds are assumed to be conformally flat.
In  \cite{Rad}, the hypothesis of conformal flatness is dropped and a 
geometrical compatibility condition is given, relating the (conformal) Killing tensor  $\mathbf K$, the Cotton-York tensor $\mathbf{A}$ and the Weyl tensor $\mathbf{C}$.
If the condition is satisfied, then a constructive method allows to write a (conformal) symmetry of the conformally invariant Laplacian (\ref{dc}).
Moreover, all (conformal) Killing tensors defined on a conformally flat manifold satisfy this condition and this explains the quantum corrections proposed in
\cite{He}.

%
%

In the present paper, instead, we look for a pair of additional functions $E$ and $E_K$
that have to be added to the Laplacian and to
the  Laplace-Beltrami quantization  $\Delta_K$ of a quadratic in the momenta first integral $K$
in order to get two commuting differential operators.
These functions, called quantum corrections, are in principle both considered as unknown to be determined. The quantization procedure obtained in this way will be called modified Laplace-Beltrami (MLB) quantization.

A compatibility condition involving both $E$ and $E_K$ is established and a consequent integrability condition involving only $E$ is derived. If we assume that $E$ is the scalar curvature, then this integrability condition reduces
to the one proposed in  \cite{Rad}. In this case, the modified Hamiltonian quantum operator coincides with the conformally invariant Laplacian (\ref{dc}).
By leaving $E$ arbitrary, instead, we have the possibility to find suitable quantum corrections for non-conformally flat metrics also (see example \cite{Rad} and \cite{KKMncf} in four dimension). 
As instance, a more geometric interpretation is given for the quantum correction given in \cite{KKMncf} and we explicitly gives commuting operators up to the five dimensional case. 

Our analysis is carried on for any quadratic constants of motion of any natural Hamiltonian and we obtain necessary and sufficient conditions for the existence of quantum corrections for the  simultaneous quantization of all the quadratic constants of motion of the
classical system.

In Section 2 we introduce Laplace-Beltrami quantization and we obtain our  main results: Theorems \ref{T1}, \ref{propB1} and \ref{SepQC}.
 
Through Theorem \ref{T1}, we characterize the quantum corrections of the Laplace-Beltrami operator and of its second-order symmetry operator involving the Ricci tensor, the Killing two-tensor associated with the quadratic first integral, and the Levi-Civita connection. 

Theorem \ref{propB1} characterizes the quantum corrections necessary to the simultaneous modified quantization of several quadratic in the momenta first integrals of the same Hamiltonian.

In Theorem \ref{SepQC}, we restrict  the previous results to the  important case of St\"ackel systems, the $N$-dimensional Hamiltonian systems admitting $N-1$ quadratic in the momenta first integrals in involution, plus  the Hamiltonian itself,  associated with the existence of coordinate systems allowing the integration by separation of variables of the Hamilton-Jacobi equation. The Laplace-Beltrami quantization of these systems has been considered in \cite{BCRI, BCRII} and we recall the main results. We address here to the case when the LB operator has to be modified by some quantum correction, for example due to existence of additional quadratic constants of motion  not included into the involutive set, a situation  typical of superintegrable systems. We give  necessary  conditions for the quantum correction allows the simultaneous modified quantization of the whole St\"ackel system.

In Section 3, we apply the results of Section 2 to three examples. The first two, taken from \cite{Rad}, show that the initial freedom in the choice of the quantum correction allows the quantization of first integrals found to be not associated with symmetries of the conformally invariant Laplacian in \cite{Rad}. In the third example, inspired by \cite{KKMncf}, we apply our technique to a hierarchy of $N$-dimensional Hamiltonian systems, all maximally superintegrable with $2N-1$																																														  constants of motion all quadratic in the momenta, defined (for $N>4$) on non-conformally-flat  manifolds. We  discuss the geometric meaning of some of the  quantum corrections considered in \cite{KKMncf}.

\section{Modified Laplace-Beltrami quantizations}
In order to study the modified Laplace-Beltrami (MLB) quantizations, we generalize here some results obtained in \cite{BCRII} about Laplace-Beltrami quantization. 
In \cite{BCRII} Laplace-Beltrami quantization of quadratic functions on cotangent bundles of Riemannian and pseudo-Riemannian manifolds was considered. Theorem 2.2 states that, given the quadratic in the momenta functions 
$$H=\frac 12 g^{ab}p_ap_b+W,\quad   K=\frac 12 K^{ab}p_ap_b+W_K,$$
 and the operators 
 \begin{equation}\label{qw}
 \widehat H=-\frac {\hbar^2}2 \Delta+W, \quad \widehat K=-\frac {\hbar^2}2 \Delta_K +W_K,
 \end{equation}
 with 
$$
\Delta_K=\nabla_a K^{ab} \nabla_b,
$$ 
then $[\widehat H,\widehat K]=0$ (i.e., $\widehat H$ and $\widehat K$ commute) if and only if  the 2-tensor $\mathbf{K}$ of  contravariant components $(K^{ab})$  is a Killing tensor for the metric $\mathbf g$  and  
\begin{equation}\label{e0}
\mathbf K\nabla W-\nabla W_K+\frac {\hbar ^2}6\delta \mathbf C_K=0,
\end{equation}
where $$\mathbf C_K=\mathbf K\mathbf R-\mathbf R\mathbf K,$$ $\mathbf R$ is the Ricci tensor and $(\delta C_K)^a=\nabla_bC_K^{ba}$ (i.e. $\delta$ is the divergence operator). It is well known that $\{H,K\}=0$ ($H$ and $K$ Poisson-commute) if and only if $\mathbf K$ is a Killing tensor and $\mathbf K\nabla W-\nabla W_K=0$, therefore, Proposition 2.5 of \cite{BCRII} states that, for any  quadratic first integral $K$ of $H$,
$$
[\widehat H,\widehat K]=0 \Leftrightarrow \delta \mathbf C_K=0.
$$ 

\begin{defi} We say that $\mathbf K$ satisfies the {\it Carter condition} if $\delta \mathbf C_K=0$. 
	\end{defi}

\begin{rmk}
\rm
The Laplace-Beltrami quantization (\ref{qw}) provides  (formally) self-adjoint differential operators acting on functions (see \cite{BCRII}). The modification of the Laplacian by an additional scalar term does not affect self-adjointness.
\end{rmk}

\begin{rmk}
\rm
If we are dealing with a constant-curvature configuration manifold or, more generally, with any Einstein manifold, where the Ricci tensor is a multiple of the metric, then the Laplace-Beltrami quantization produces symmetry operators for the quantum Hamiltonian operator for each quadratic first integral of the classical Hamiltonian.

\end{rmk}

\begin{rmk}\label{rd}
\rm
 If  the Killing tensor $\mathbf K$ and the Ricci tensor $\mathbf R$ are simultaneously diagonalized in in some coordinate system, then the Carter condition is automatically satisfied, because $\mathbf{C}_K=0$.
 \end{rmk}

We adapt the previous results to the case of the modified quantization. We have

\begin{teo}\label{T1}
Let $\Delta_E$ be the modified Laplacian $\Delta_E=\Delta+E$, with $E$ any given scalar, let be $\widehat H_E=-\frac {\hbar^2}2 \Delta_E+V=-\frac {\hbar^2}2 (\Delta +E)+V$ and $\widehat K_{E_K}=-\frac {\hbar^2}2 (\Delta_K +E_K)+V_K$, where $E_K$ is a scalar to be determined, then
$$
[\widehat H_E,\widehat K_{E_K}]=0 
$$ 
if and only if 
\begin{enumerate}
\item the 2-tensor $\mathbf{K}$ of  contravariant component $(K^{ab})$  is a Killing tensor,

\item the following equation holds
\begin{equation}\label{e2}
\frac {\hbar ^2}6 \delta \mathbf C_K+ \frac {\hbar^2}2 \left(\nabla E_K-\mathbf K\nabla E\right)+\mathbf K\nabla V-\nabla V_K=0.
\end{equation}
\end{enumerate}
\end{teo}

\proof The modified operators $\widehat H_E$ and $\widehat H_{K_E}$ correspond to the  Laplace-Beltrami quantization (\ref{qw}) of the Hamiltonians $H_E=\frac 12 g^{ab}p_ap_b+W$, $K_{E_K}=\frac 12 K^{ab}p_ap_b+W_K$, where
\begin{eqnarray*}
W&=&V-\frac {\hbar^2}2 E,\\
W_K&=&V_K-\frac {\hbar^2}2 E_K,
\end{eqnarray*}
and we may apply Theorem 2.2 of \cite{BCRII}.
By substituting these relations in (\ref{e0}) we obtain
$$
\mathbf K(\nabla V-\frac {\hbar^2}2\nabla E)-\nabla V_K+\frac {\hbar ^2}2\nabla E_K+\frac {\hbar^2}6\delta \mathbf C_K=0,
$$
that gives equation (\ref{e2}).\qed

Since  $H_K$ is a first integral of $H$ if and only if $\mathbf K$ is a Killing tensor and
$$
\mathbf K\nabla V-\nabla V_K=0,
$$
it follows,
\begin{prop}\label{P1}
If $\{H,K\}=0,$ then $[\widehat H_E,\widehat K_{E_K}]=0$ if and only if
\begin{equation}\label{se}
\nabla E_K-\mathbf K\nabla E+ \frac 13 \delta \mathbf C_K=0.
\end{equation}
\end{prop}
From Proposition 2,  integrability conditions for $E_K$ can be immediately obtained 
\begin{prop}\label{P1bis}
The integrability conditions of $E_K$ coincide with the symmetry of the 2-tensor  $\mathbf P=\nabla (\mathbf K\nabla E-\tfrac{1}{3}\delta \mathbf C_K)$
of covariant components
\begin{equation}\label{icS}
P_{ab}=\nabla_a [(K\nabla E)_b-\tfrac 13 (\delta C_K)_b].
\end{equation}
Equivalently, 
\begin{equation*}d(\mathbf K\nabla E-\frac 13 \delta \mathbf C_K)^\flat=0.
\end{equation*}
\end{prop}

\begin{prop}
	Let $(E,E_K)$ be a pair of functions satisfying condition (\ref{se}). The most general solutions $(E_g,\, E_{Kg})$ of (\ref{se}) are of the form
	\begin{equation*}
	E_g=E+E_o, \qquad E_{Kg}= E_K+ E_{Ko},
	\end{equation*}
	with $(E_o,E_{Ko})$ belonging to the linear
	space defined by the solutions of  the linear homogeneous differential equation 
	\begin{equation}\label{seo}
	\mathbf{K}\nabla E_o =\nabla E_{Ko},
	\end{equation}
	whose integrability condition is $d(\mathbf{K}\nabla E_o)^\flat=0$.
\end{prop}
\begin{proof}
	Equation (\ref{se}) is a linear non-homogeneous
	differential condition in $E$ and $E_K$, whose associated homogeneous equation is (\ref{seo}).
\end{proof}

Condition (\ref{seo}) is equivalent to $\{H+E_o,K+E_{Ko}\}=0$.



When a Hamiltonian $H$ admits several quadratic first integrals and  we know the quantum corrections allowing the MLB quantization of $H$ together with each one of the first integrals, we may ask ourselves if there exists a single quantum correction of $H$ allowing the simultaneous MLB quantization of all those first integrals. We have

\begin{teo}\label{propB1}
	Let $H$ be a Hamiltonian with $k$ quadratic in the momenta first integrals $K^{(i)}$.
	Let us assume that for each $i$ there exist quantum corrections $E^{(i)}$, $F^{(i)}$ such that $$[\widehat{H}_{E^{(i)}},\widehat{K}^{(i)}_{F^{(i)}}]=0.
	$$
	There exists  a simultaneous quantum correction $E$ of $H$ for all the    $K^{(i)}$ if and only if
	$$
	d(\mathbf K^{(i)}\nabla E^{(i)})=d(\mathbf K^{(i)}\nabla E), \quad i=1,\ldots,k.
	$$
	The quantum corrections $F^{(i)}$ must be replaced by $W^{(i)}=F^{(i)}+ \bar{F}^{(i)}$ where
	each $\bar{F}^{(i)}$ is a potential of the closed one-form
	$\mathbf K^{(i)}\nabla (E-E^{(i)})$:
	$$[\widehat{H}_{E},\widehat{K}^{(i)}_{W^{(i)}}]=0.$$
\end{teo}
\begin{proof} From Proposition \ref{P1}, a simultaneous quantum correction $E$ exists if and only if it satisfies
	$$
	\frac 13 \delta \mathbf C_{K^{(i)}}=\mathbf K^{(i)}\nabla E-\nabla W^{(i)}, \quad i=1,\ldots,k,
	$$
	for suitable functions $W^{(i)}$. From the same Proposition we have
	$$
	\frac 13 \delta \mathbf C_{K^{(i)}}=\mathbf K^{(i)}\nabla E^{(i)}-\nabla F^{(i)}, \quad i=1,\ldots,k,
	$$
	then, the first equation becomes
	
	\begin{equation}\label{P1s}
	\mathbf K^{(i)}\nabla (E^{(i)}- E)=	\nabla (F^{(i)}- W^{(i)}).
	\end{equation}
	The integrability condition of this equation is
	\begin{equation}
	d(\mathbf K^{(i)}\nabla (E^{(i)}- E))=d(\mathbf K^{(i)}\nabla E^{(i)})-d(\mathbf K^{(i)}\nabla E)=0.
	\end{equation}

\end{proof}

We remark that the previous statement  does not consider the possible commutation relations among the first integrals, required by Liouville and quantum integrability of the Hamiltonian.

The St\"ackel systems are natural  $N$-dimensional Hamiltonian systems with $N$ quadratic independent  first integrals in involution,  such that the Killing tensors  associated with these first integrals are simultaneously diagonalized in orthogonal coordinates, coordinates determined by the eigenvectors of the Killing tensors. These coordinates are called St\"ackel coordinates. It follows \cite{BCRI} that St\"ackel systems are Liouville integrable, one of the first integrals can be chosen as the Hamiltonian $H$ itself. The characteristic property of St\"ackel systems is that the Hamilton-Jacobi equation of $H$ is additively separable in these orthogonal coordinates. We call Killing-St\"ackel algebra the linear space generated by the Killing tensors of a St\"ackel system.

It is known (\cite{BCRI, BCRII} and references therein) that the Schr\"odinger equation is multiplicatively separable in St\"ackel coordinates if and only if the coordinates diagonalize the Ricci tensor also. In this case, the condition $C_K=0$ is verified for all the elements of the Killing-St\"ackel algebra.

Moreover, St\"ackel systems are LB quantizable provided the Carter condition holds for the Killing-St\"ackel algebra. We recall that, for  a St\"ackel system \cite{BCRI}
\begin{equation}
\delta \mathbf C_{K_i}=0, \;i=1, \ldots, N, \Leftrightarrow [\widehat K_i,\widehat K_j]=0, \; i,j =1,\ldots, N.
\end{equation}   
Consequently, in this case no quantum correction is necessary. However, if the St\"ackel system is superintegrable, quantum corrections can be necessary to quantize the additional first integrals. 

If the Carter condition holds for all the quadratic Killing 2-tensors $(\mathbf K_i)$, $i=1,\ldots,N$, of a Killing-St\"ackel algebra, we call it {\it pre-Robertson condition}. If, in addition, $\mathbf C_{K_i}=0$, $i=1,\ldots,N$, we say that $(\mathbf K_i)$ satisfy the {\it Robertson condition} (see \cite{Car, BCRI} and references therein).

 Since the $\mathbf K_i$ are simultaneously diagonalized, the Robertson condition is equivalent to the diagonalization of the Ricci tensor in St\"ackel coordinates.  As shown in  \cite{BCRII}, the satisfaction of the pre-Robertson condition allows the LB quantization of the Killing-St\"ackel algebra in any manifold.

The LB quantization of St\"ackel systems, discussed in \cite{BCRI, BCRII}, provides a Schr\"odinger operator that admits orthogonal (multiplicative) separation of variables in the same orthogonal coordinates allowing the (additive) separation of variables for the Hamilton-Jacobi equation of the original St\"ackel system, only if the Robertson condition is satisfied.

In \cite{Bl},  an  alternative quantization procedure of St\"ackel systems is considered and applied in \cite{Bl1}  to a class of superintegrable St\"ackel systems with all quadratic in the momenta constants of motion.  The procedure gives separable, and superintegrable, Schr\"odinger operators even when the Robertson condition is not verified. This result is obtained by replacing the metric tensor $\mathbf g$ defining the Laplace-Beltrami operator with a suitable   conformal deformation $\bar{\mathbf g}$. An analysis of the similarities and differences  between these two approaches is worthwhile, but beyond the scope of the present paper.

From Theorem \ref{T1} and Theorem \ref{propB1} we derive the following result, useful for the quantization of superintegrable St\"ackel systems.

\begin{teo} \label{SepQC} Let $(H, K_1 \ldots, K_{N-1})$ be a St\"ackel system, of Hamilton function $H$ and separable coordinates $(q^i)$, satisfying the pre-Robertson condition.  Let  $\mathbf K_i$ be the Killing tensors such that $K_i=\frac 12 K_i^{lm}p_lp_m+V_i$. Let $\widehat H_E$ be the quantum Hamiltonian obtained from $H$ with any  quantum correction $E$. If 
$$
d(\mathbf K_idE)=0,\quad  i=1, \ldots, n-1,
$$
then, the system $( H-\frac{\hbar^2}{2}E, K_{1}-\frac{\hbar^2}{2}E_{1} \ldots, K_{N-1}-\frac{\hbar^2}{2}E_{N-1} )$, where $dE_i=\mathbf K_i dE$, is a St\"ackel system separable in $(q^i)$. Moreover,  by denoting by $\widehat K_{E_i}$ the modified operators with quantum corrections $E_{i}$,  the system $(\widehat H_E, \widehat K_{E_1}, \ldots, \widehat K_{E_{N-1}})$ is a quantum integrable system. Furthermore, if the stronger Robertson condition is satisfied by the St\"ackel system, then  the quantum system $(\widehat H_E, \widehat K_{E_1}, \ldots, \widehat K_{E_{N-1}})$ is also separable in $(q^i)$.
\end{teo}

\proof  Since the pre-Robertson condition is satisfied,  the  equation (\ref{se}) of Proposition \ref{P1} reads as
\begin{equation}\label{QQ}
\nabla E_{K_i}= \mathbf K_i\nabla E=0, \quad i=1,\ldots,N-1,
\end{equation}
whose integrability conditions are
$$
d (\mathbf K_i d E)=0.
$$
Since the $\mathbf K_i$ are Killing tensors,  equation (\ref{QQ}) means that $$\{H-\tfrac{\hbar^2}{2}E,K_i-\tfrac{\hbar^2}{2}E_{i}\}=0,$$ and, in separable coordinates $(q^j)$ for the  St\"ackel system $(H, K_1 \ldots, K_{N-1})$, that $E=\phi_j(q^j)g^{jj}$, i.e.  $E$ is a St\"ackel multiplier \cite{BCRI}, as well as $-\frac{\hbar^2}{2}E$. Therefore, also the commutation conditions  $$\{H_i-\tfrac{\hbar^2}{2}E_{i}, H_j-\tfrac{\hbar^2}{2}E_{j}\}=0$$ follow and  $( H-\frac{\hbar^2}{2}E, K_{1}-\frac{\hbar^2}{2}E_{1} \ldots, K_{N-1}-\frac{\hbar^2}{2}E_{N-1} )$ is again a St\"ackel system. Since the Killing tensors $\mathbf K_i$ and the Ricci tensor are unchanged, the pre-Robertson condition still holds true for the new system and the separable coordinates are the same. Therefore, 
 the quantized system obtained from it, that coincides with  $(\widehat H_E, \widehat K_{E_1}, \ldots, \widehat K_{E_{N-1}})$,  is quantum integrable. If the Robertson condition $\mathbf{C}_{K_i}=0$, $i=1,\ldots,N-1$, holds for the original system, it holds clearly for the system modified by the potentials $-\frac{\hbar^2}2E_i$, then the modified quantum system is also separable in $(q^i)$ \cite{BCRI}. \qed


\section{Examples}
We show several detailed applications of the modified quantization. The following examples are
inspired by the papers  \cite{Rad} and \cite{KKMncf}	
The first two examples show how the freedom in the choice of the quantum correction for the Hamiltonian allows to construct commuting 
symmetry operators:  two geodesic Hamiltonians admitting a quadratic first integral considered in \cite{Rad} are analysed. 
In \cite{Rad} these examples where used  to show how to apply the compatibility condition
between the metric and the (conformal) Killing tensor. This condition means that
a specific one-form, called obstruction,  
 is closed:
\begin{equation}\label{Radcom}
d\left(
(
C_{\ st}^{r\ a}\nabla_r 
- 3 A^{\ a}_{st})K^{st}g_{ab}dq^b\right)=0,
\end{equation}
where $K^{ab}$ are the contravariant components
of the (conformal) Killing tensor associated with the first integral, while  $\mathbf{A}$ and $\mathbf{C}$ are  the Cotton-York tensor and the Weyl tensor, respectively.
The scalar function $f$, local potential of the 
above closed one-form, has to be added to the standard conformally-equivariant quantization of the quadratic first integral in order to ensure the commutation with the conformally invariant Laplacian.
When the obstruction is not closed, the first integral is not assocated with a second order symmetry of the conformally invariant Laplacian.
Nevertheless,  as shown in the examples below, a different scalar $E$ in the modified quantization could be used to create a different quantum Hamiltonian with a second-order symmetry operator. 

{

\subsection{Example 1}
Let us consider the following three-dimensional St\"ackel metric with an ignorable coordinate $q^1$
\begin{equation}\label{rad1}
H=\frac{1}{2}p_1^2+\frac{1}{u(q^2)+v(q^3)}\left(\frac{1}{2}p_2^2+\frac{1}{2}p_3^2\right),
\end{equation}
and its quadratic first integral
\begin{equation}
K=\frac{1}{u(q^2)+v(q^3)}\left(v(q^3)p_2^2-u(q^2)p_3^2\right).
\end{equation}
Since the Ricci tensor is diagonal for all functions $u$ and $v$, by Remark \ref{rd}, the ordinary Laplace-Beltrami quantization produces commuting quantum operators and no quantum correction is necessary, i.e., $E=0$ and $E_K=0$ satisfy the compatibility condition (\ref{se}).

On the contrary, since the metric is not conformally flat for arbitrary functions $u$ and $v$, we see that the quantum correction given by the scalar curvature
\begin{equation*}E= \Sc=\frac{1}{u(q^2)+v(q^3)}(\partial^2_{q^2}+\partial^2_{q^3})\ln(u(q^2)+v(q^3)),
\end{equation*}
is not compatible with $\mathbf K$. Indeed, compatibility condition (\ref{icS}) reduces to $d\mathbf{K}dE=0$ which is satisfied if and only if
$$(\partial^2_{q^2}+\partial^2_{q^3})\partial_{q^2}\partial_{q^3}\ln(u(q^2)+v(q^3))=0.$$
This analysis is in accordance with the results of \cite{Rad} where it is shown that $\mathbf K$ is not associated with a symmetry of the conformally invariant Laplacian, when the above equation is not satisfied.
\subsection{Example 2}
Let us consider the following geodesic Hamiltonian  \begin{equation}\label{er2}
H=\frac{1}{2}p_r^2+\frac{1}{2}p_z^2+
\frac{1}{2}\left(\frac{1}{r^2}-\frac{a^2}{z^2}\right)p_\phi^2,
\end{equation}
obtained as a reduction of the Minkowski metric on $\mathbb{M}^4$ along the Killing vector $\mathbf X$ defined in pseudo-Cartesian coordinates $(x^0,x^1,x^2,x^3)$ as
$\mathbf X=(x^1\partial_0+x^0\partial_1)+a(x^1\partial_2-x^2\partial_1)$, with $a \in \mathbb{R}-\{0\}$
(see \cite{Rad} and references therein).
The metric of $H$ is conformal to the metric tensor of (\ref{rad1}) with $u(r)=1/r^2$ and $v(z)=-a/z^2$ and it is again a St\"ackel metric. Indeed, 
a metric conformal to a St\"ackel metric is of St\"ackel
type if and only if the conformal factor is of the form $\sum_i f_i(q^i)g^{ii}$ (see \cite{BCRHJ0}). In this example
 the conformal factor relating the two metrics is  $g^{22}$.
It is easy to check directly that the function
\begin{equation}
K=\frac{1}{2}p_r^2+\frac{1}{2r^2}p_\phi^2,
\end{equation}
is a quadratic first integral of $H$.
The Ricci tensor of the metric of (\ref{er2}) is not diagonal in $(r, \phi,z)$ and the covariant form of $\delta \mathbf{C}_K$ is
$$
(\delta \mathbf{C}_K)^\flat=-\frac{3a^2}{a^2 r^2-z^2}\big(r(2a^2r^2+3z^2)dr+z(2a^2z^2+3r^2)dz\big), 
$$
that is different from zero.
 
However, the quantum correction equation (\ref{se}) can be solved
for
\begin{equation}\label{ER}
E=-\frac{1}{8}\Sc-\frac{(a+1)^2}{4}\frac{a^2 r^2+z^2}{a^2 r^2-z^2},
\end{equation}
$$
E_K=-\frac{a^2(a^2r^2+4z^2)}{4(a^2r^2-z^2)^2},
$$
where 
$$\Sc=\frac{6a^2(r^2+z^2)}{(a^2r^2-z^2)},$$ 
is the scalar curvature of the metric tensor of (\ref{er2}).
Also in this case we recover the fact stated in \cite{Rad} that $\mathbf K$ is not a symmetry of the conformally invariant Laplacian. Indeed, the  quantum correction (\ref{ER}) requires an additional term to $-\Sc/8$ (the correction which transforms a Laplacian in the conformally invariant Laplacian in a three-dimensional manifold) in order to satisfy the compatibility condition (\ref{icS}).
}

\subsection{Example 3}

In \cite{KKMncf} the authors introduce a class of $N$-dimensional Hamiltonians, based upon the Tremblay-Turbiner-Winternitz system,  that are superintegrable St\"ackel systems, therefore admitting $N$ quadratic first integrals associated with separation of variables, and $N-1$ other independent first integrals. These last constants of motion are polynomial in the momenta of degree depending on a rational parameter.
These Hamiltonians are  defined in non-conformally-flat Riemannian manifolds for $N>3$, thus, their
quantization requires quantum corrections in order to preserve the superintegrability of the corresponding quantum systems. In \cite{KKMncf} is explicitly discussed the $four$-dimensional case and it is found that the quantum correction is a linear combination of the scalar curvature with an additional scalar determined from the Weyl tensor.

Here, we review the example of \cite{KKMncf} step by step, under the assumption that also the other $N-1$ first integrals are quadratic in the momenta. Given the Hamiltonian 
\begin{eqnarray}\label{KMH}
	H = L_4 = p^2_r + \alpha r^2 +\frac {L_3}{r^2}, \\
	L_3 = p^2_{\theta_1}+
	\frac{\beta_1}{\cos^2(k_1\theta_1)}+
	\frac{L_2}{\sin^2(k_1\theta_1)},\\
	L_2 = p^2_{\theta_2}+
	\frac{\beta_2}{\cos^2(k_2\theta_2)}+
	\frac{L_1}{\sin^2(k_2\theta_2)},\\
	L_1 = p^2_{\theta_3} +\frac{\beta_3}{\cos^2(k_3\theta_3)}+
	\frac{\beta_4}{\sin^2(k_3\theta_3)},
\end{eqnarray}
we examine the classical and quantum second-order superintegrability of all $n$-dimensional subsystems ($n=2,\ldots,4$) with $\beta_1=\ldots=\beta_4=\alpha=0$ (we do not loose in generality since equation (\ref{se}) is unaffected by the addition of scalar potentials to the classical system) and determine the quantum corrections.

In the following, we denote by $H_i$ the quadratic first integrals associated with separation of variables, that are all in involution, and by $K_i$ the other quadratic first integrals.

The Hamiltonian (\ref{KMH}) admits three additional independent second-degree first integrals for $k_l=2^l$, $l=1,\ldots,3$. Indeed,
by setting $q^{4-j}=k_j\theta_j,$ $p_{4-j}=p_{\theta_j}/k_j$, $j=1,\ldots,3$, $q^4=r$, $p_4=p_r$, the Hamiltonian $H_4=H/2$ becomes 
\begin{equation}\label{H4}
H_4=\frac 12 p_4^2+ \frac{4}{(q^4)^2}\left(\frac 12 p_3^2+\frac 4{\sin ^2 q^3}\left( \frac 12 p_2^2+\frac 4{\sin ^2 q^2}\left(\frac 12 p_1^2\right)\right)\right).
\end{equation}
{It can be proved that the $H_i$ can be obtained recursively from $H_1=p_1^2/2$ by applying the procedure of extension (see \cite{CDRttw} and references therein). 
Roughly speaking, this procedure 	
 transforms a Hamiltonian $L$ into a Hamiltonian
 of the form
  $$H= \frac{1}{2}p_u^2+\frac{m^2}{n^2}\alpha(u)L+f(u),$$
   with
one more degree of freedom and admitting a polynomial  first integral $K$ of degree depending on $\frac mn$ computable through a recursive procedure. For $\frac mn=2$,  $K$ is quadratic. (for more about the extension procedure, see \cite{CDRfi, CDR, CDRfrac, CDRttw, CDR1} ) We leave for a next paper the details of the construction of suitable quantum corrections for a general extension with quadratic first integrals.
}

\noindent
\textbf{The two-dimensional system.}
The Hamiltonian
\begin{equation}\label{H2}
H_2=\frac 12 p_2^2+\frac 4{\sin ^2 q^2}\left(\frac 12 p_1^2\right),
\end{equation}
is defined on a two-dimensional constant curvature manifold and admits (together with $H_2$ itself) the first integrals
\begin{equation}\label{H1}
H_1=\frac 12 p_1^2,
\end{equation}
\begin{equation}\label{K2}
K_2=\cos (q^1) p_2^2-4\sin (q^1) \tan ^{-1} (q^2) p_1p_2-4 \cos (q^1) \tan^{-2} (q^2) p_1^2.
\end{equation}
Both $H_1$ and $K_2$ are associated with commuting operators through the standard minimal quantization. Indeed, any constant curvature manifold is an Einstein manifold and therefore we have $\mathbf{C}_K=0$ for any symmetric Killing two-tensor $\mathbf{K}$.

\noindent
\textbf{The three-dimensional system.}
Let us consider the three-dimensional Hamiltonian
\begin{equation}\label{H3}
H_3= \frac{1}{2}p_3^2+\frac{4}{\sin^2 q^3}H_2,
\end{equation}
which is defined on a conformally flat manifold with non constant curvature
$$
\Sc_3=-6\left(1+\frac{1}{\sin^2 q^3}\right),
$$
with  Ricci tensor diagonal in these coordinates. 
Four additional  first integrals are (\ref{H1}), (\ref{H2}), (\ref{K2}) 
 and 
 \begin{equation}\label{K3}
 K_3=\cos (q^2) p_3^2-4\sin (q^2) \tan ^{-1} (q^3) p_2p_3-8 \cos (q^2) \tan^{-2} (q^3) H_2.
 \end{equation} 
The Killing tensors $\mathbf{H}_1$, $\mathbf{H}_2$, $\mathbf{K}_2$ associated with  $H_1$, $H_2$ and $K_2$ satisfy the Carter condition and would not need quantum corrections, while
the Killing tensor associated with $K_3$ does not  satisfy the Carter condition and therefore a quantum correction is needed.  
Because of the conformal flatness, all  the  Killing tensors  associated with the quadratic in the momenta first integrals satisfy equation (\ref{icS})  of Proposition 2, with $E=-\frac{1}{8}\Sc_3$, in agreement with \cite{Rad}.
Equation (\ref{se})  is satisfied by the following quantum corrections for $\widehat{H}_3$ and $\widehat{K}_3$ respectively 
\begin{equation}\label{E3-3d}
E_3=E(H_3,K_3)= -\frac{1}{8}\Sc_3=-\frac{3}{4}\left(1+\frac{1}{\sin^2 q^3}\right),
\end{equation}
\begin{equation} \label{EK3-3d}
E_{K_3}=E(K_3,H_3)=-\frac{1}{2}\cos q^2(5 \tan^{-2}(q^3) +2).
\end{equation}
Now, in order to get simultaneous quantization for all first integrals, according to Theorem \ref{SepQC},
we need to check that $E_3$ satisfies 
$$d\mathbf{H}_1 dE_3=d\mathbf{H}_2 dE_3=d\mathbf{K}_2 dE_3=0,$$
 for the Killing tensors $\mathbf{H}_1$, $\mathbf{H}_2$, $\mathbf{K}_2$. However, since $E_3$ depends on $q^3$ only and the first integrals of $H_2$ do not involve $p_3$, we have  
$$
\mathbf{H}_1 dE_3=\mathbf{H}_2 dE_3=\mathbf{K}_2 dE_3=0,
$$ 
and, therefore, no quantum corrections are needed for $\widehat{H}_2$, $\widehat{H}_1$, and $\widehat{K}_2$. 
Thus, we get that
$$
[\widehat{H}_3-\tfrac{\hbar^2}{2}E_3,\widehat{K}_3-\tfrac{\hbar^2}{2}E_{K_3}]= [\widehat{H}_3-\tfrac{\hbar^2}{2}E_3,\widehat{K}_2]=0,
$$
$$
[\widehat{H}_3-\tfrac{\hbar^2}{2}E_3,\widehat{H}_2]= [\widehat{H}_3-\tfrac{\hbar^2}{2}E_3,\widehat{H}_1]=0.
$$

\noindent
\textbf{The four-dimensional system.}
The four-dimensional Hamiltonian (\ref{H4}) has the additional first integral 
\begin{equation}\label{K4}
K_4=\cos (q^3) p_4^2-4\sin (q^3) (q^4) ^{-1}  p_3p_4-8 \cos (q^3) (q^4)^{-2}  H_3,
\end{equation}
and the five constants of motion (\ref{H2}--\ref{K3}) inherited from the three-dimensional Hamiltonian.
Its underlying manifold it is not conformally flat, since the Weyl  tensor does not vanish, but has diagonalized Ricci tensor in the $(q^i)$ (recall that this means that the Robertson condition is satisfied, hence $(q^i)$ are separable coordinates also for the Laplace-Beltrami operator). Following \cite{KKMncf} we introduce the Weyl scalar
$$
W_4=\sqrt{3 W_{abcd}W^{abcd}}=\frac{24}{(q^4)^2\sin^2(q^3)}.
$$
The quantum corrections for quantizing $H_4$ and $K_4$ (as we will show elsewhere in a more general context) are
$$
E_4=E(H_4,K_4)=-\frac{1}{6}\Sc_4= \frac{4+3\sin^2 (q^3)}{(q^4)^2\sin^2(q^3)},
$$

$$
E_{K_4}=E(K_4,H_4)= -2\cos q^3 \,\frac{4+3\sin^2(q^3)}{(q^4)^2\sin^2(q^3)},
$$
while a pair of quantum corrections allowing simultaneous quantization for 
$H_4$ and $K_3$ is
$$
\bar{E}_4=E(H_4,K_3)=\frac{3}{(q^4)^2}\left(1+\frac{1}{\sin^2 q^3}\right),
$$
$$
\bar{E}_{K_3}=E(K_3,H_4)=-\frac{1}{2}\cos q^2(5 \tan^{-2}(q^3) +2).
$$
We remark that $\bar{E}_{K_3}$ coincides with (\ref{EK3-3d}), while $\bar{E}_4$ is the product of (\ref{E3-3d}) by $-4(q^4)^{-2}$, which is the coefficient of $H_3$ in its extension $H_4$. 
All other Killing tensors satisfy Carter condition and do not need quantum correction.
Hence, according to Theorem \ref{propB1}, in order
to get simultaneous modified quantization
we look for a function $E$ such that
$$
d\mathbf{K}_4 d(E_4-E)=d\mathbf{K}_3 d(\bar{E}_4-E)=
0,
$$
and 
$$
d\mathbf{K}_2 dE=0, \qquad d\mathbf{H}_i dE=0, \qquad (i=1,\ldots, 3).
$$
By solving the first two PDEs for $E$ with the assumption that $E$ does not depend on $q^1$, we get
$$
E=\frac{3}{(q^4)^2}+\frac{3}{(q^4)^2\sin^2(q^3)}+
\frac{C_1+C_2 \cos(q^2)}{(q^4)^2\sin^2(q^3)\sin^2(q^2)}+
C_3(q^4)^2+ C_4,
$$
where $C_1,\ldots,C_4$ are arbitrary constants.
We observe that for $C_i=0$ we get the quantum correction used in \cite{KKMncf} which is a linear combination of the scalar curvature and of the Weyl scalar.
By imposing the further condition $d\mathbf{K}_2dE=0$
we get
$$
(\cos^4 q^2-6\cos^2 q^2-3)C_2 -8C_1\cos q^2=0.
$$
Thus, we need $C_1=C_2=0$. Furthermore, for this choice of the constants $C_i,$
we have $\mathbf{H}_1dE=\mathbf{H}_2dE=\mathbf{K}_2dE=0$ and
$d\mathbf{H}_3dE=0$.
Hence, the required quantum correction for $\widehat{H}_4$ is
\begin{equation*}
E=\frac{3}{(q^4)^2}+\frac{3}{(q^4)^2\sin^2(q^3)}+
C_3(q^4)^2+ C_4.
\end{equation*}
The modified Hamiltonian 
operator ($C_4\in\mathbb{R}$)
\begin{equation}
\widehat{H}_4-\frac{\hbar^2}{2}\left(3\frac{1+\sin^2 q^3}{(q^4)^2\sin^2 q^3}+
C_4(q^4)^2\right),
\end{equation}
commute with  the six following modified operators:
\begin{eqnarray}
&&\widehat{K}_4-{\hbar^2}\cos q^3\left(
C_4 (q^4)^2- 3\frac{1+\sin^2 q^3}{(q^4)^2\sin^2 q^3}\right), \\
&&\widehat{K}_3 +\frac{\hbar^2}{4}\cos q^2(5 \tan^{-2}(q^3) +2), \label{K3c}\\
&&\widehat{K}_2, \\
&&\widehat{H}_3-\frac{3\hbar^2}{8\sin^2 q^3 }, \label{H3c}\\
&&\widehat{H}_2, \\
&&\widehat{H}_1. \label{H1c}
\end{eqnarray}

\begin{rmk} \rm
Operators (\ref{K3c}--\ref{H1c}) are the five symmetry operators of (\ref{H3c}) determined in the three-dimensional case.
\end{rmk}

\noindent
\textbf {The 5-dimensional system.}
 As a further step we consider a Hamiltonian $H_5$ with five degrees of freedom.
In this case, to apply the extension procedure, the four dimensional Hamiltonian (\ref{H4}) must be replaced
by
\begin{equation}\label{H4b}
{H'}_4=\frac 12 p_4^2+ \frac{4}{\sin^2 q^4}\left(\frac 12 p_3^2+\frac 4{\sin ^2 q^3}\left( \frac 12 p_2^2+\frac 4{\sin ^2 q^2}\left(\frac 12 p_1^2\right)\right)\right),
\end{equation}
that admits the first integrals
\begin{equation}\label{K4b}
{K'}_4=\cos (q^3) p_4^2-4\sin (q^3) \tan^{-1}(q^4)  p_3p_4-8 \cos (q^3) \tan^{-2}(q^4)  H_3,
\end{equation}
and the five constants of motion (\ref{H2}--\ref{K3}) inherited from the three-dimensional case.
The underlying manifold of $H'_4$  is not conformally flat, since the Weyl tensor is not zero, but has  Ricci tensor diagonalized in the $(q^i)$. The Weyl scalar is
$$
{W'}_4=\sqrt{3 {W'}_{abcd}{W'}^{abcd}}=\frac{24}{\sin^2(q^4)\sin^2(q^3)}.
$$
The quantum corrections for quantizing ${H'}_4$ and ${K'}_4$  (as we will show elsewhere in a more general context) are
$$
{E'}_4=E({H'}_4,{K'}_4)=-\frac{1}{6}\Sc_4= \frac{4+3\sin^2 q^3}{\sin^2 q^4 \sin^2 q^3} +2,
$$

$$
{E'}_{K_4}=E({K'}_4,{H'}_4)= -2\cos q^3 \,\left(\frac{4+3\sin^2 q^3}{\tan^2 q^4 \sin^2 q^3 }+1\right),
$$
while, a pair of quantum corrections allowing simultaneous quantization for
${H'}_4$ and $K_3$ is
$$
\bar{{E'}}_4=E({H'}_4,K_3)=\frac{3}{\sin^2 q^4}\left(1+\frac{1}{\sin^2 q^3}\right),
$$
$$
\bar{{E}}_{K_3}=E(K_3,{H'}_4)=E({K}_3,H_3)=-\frac{1}{2}\cos q^2(5 \tan^{-2} q^3 +2).
$$
By repeating similar computations of above, we find that a quantum correction allowing simultaneous quantization of $H'_4$ with all the other first integrals is given by 
$$
E'=\frac{3}{\sin^2 q^4}+\frac{3}{\sin^2 q^4 \sin^2 q^3}+
C_4\tan^2 q^4,
$$ 
with $C_4\in\mathbb{R}$.
The modified Hamiltonian operator
\begin{equation} \label{H4c}
\widehat{{H'}}_4-\frac{\hbar^2}{2}\left(\frac{3}{\sin^2 q^4}\left(1+\frac{1}{\sin^2 q^3}\right)+
C_4\tan^2 q^4\right),
\end{equation}
commute with  the modified operator
\begin{equation}\label{K4c}
\widehat{{K'}}_4-{\hbar^2}\cos q^3\left(
C_4 \tan^{-2} q^4-\tan^2 q^4 -3 - \frac{3}{\sin^2 q^3}\right), \\
\end{equation}
and with the four operators (\ref{K3c} -- \ref{H1c}).

The five dimensional extended Hamiltonian is therefore
\begin{equation}\label{H5}
{H}_5=\frac 12 p_4^2+\frac{1}{(q^5)^2}\left(\frac 12 p_4^2+ \frac{4}{\sin^2 q^4}\left(\frac 12 p_3^2+\frac 4{\sin ^2 q^3}\left( \frac 12 p_2^2+\frac 4{\sin ^2 q^2}\left(\frac 12 p_1^2\right)\right)\right)\right),
\end{equation}
with first integrals
\begin{equation}\label{K5}
{K}_5=\cos q^4 p_5^2-4\sin q^4 (q^5)^{-1}  p_4p_5-8 \cos q^4 (q^5)^{-2} {H'}_4,
\end{equation}
and the seven constants of motion of ${H'}_4$.
The underlying manifold has the scalar curvature and the Weyl scalar given by
$$
\Sc_5= -12\frac{6\sin^2 q^3 +8+3\sin^2 q^4 \sin^2 q^3 }{(q^5)^2\sin^2 q^4\sin^2 q^3},
$$
\begin{equation}\label{w5}
W_5=4\sqrt{6}\frac{\sqrt{3\sin^4 q^3 + 8 \sin^2 q^3 + 48}}{(q^5)^2\sin^2 q^4\sin^2 q^3}.
\end{equation}
We can again find a quantum correction
for all second order operators
and a simultaneous quantum correction 
\begin{equation}\label{e5}
E=C_5(q^5)^2+\frac{27}{ 4(q^5)^2}+\frac{12}{(q^5)^2\sin^2 q^4}\left(1+\frac{1}{\sin^2 q^3} \right),
\end{equation}
with $C_5 \in\mathbb{R}$,
such that $\widehat{H}_5 - \frac{\hbar^2}{2} E$ commutes with the modified operator
$$
\widehat{K}_5- {\hbar^2} \cos q^4\left(C_5(q^5)^2 -\frac{21}{ 4(q^5)^2}-\frac{12}{(q^5)^2\sin^2 q^4}\left(1+\frac{1}{\sin^2 q^3} \right)\right),
$$
with the operators (\ref{H4c}) and (\ref{K4c})  for $C_4=0$ and
with the four operators (\ref{K3c} -- \ref{H1c}).

\begin{rmk} \rm
In dimension 4 the simultaneous quantum correction for $C_4=0$ is a linear combination of $\Sc_4$ and $W_4$
(see also \cite{KKMncf}),
while in dimension 5 any simultaneous quantum correction (\ref{e5}) of $H'_4$ with all the other first integrals is not a rational function of $\Sc_5$ and $W_5$.
\end{rmk}

\begin{rmk} \rm
The construction of above can be iterated until dimension $N$, with for $i<N$

$$
H_i=\frac 12 p_i^2+\frac 4{\sin ^2(q^i)}H_{i-1},
$$
$$
K_i=\cos (q^{i-1}) p_{i}^2-4\sin (q^{i-1}) \tan ^{-1} (q^i) p_ip_{i-1}-8 \cos (q^{i-1}) \tan^{-2} (q^i) H_{i-1}, 
$$
and for $i=N$
 $$
 H_N=\frac 12 p_N^2+\frac 4{(q^N)^2}H_{N-1},
 $$
$$
K_N=\cos (q^{N-1}) p_{N}^2-4\sin (q^{N-1}) (q^N)^{-1}  p_Np_{N-1}-8 \cos (q^{N-1}) (q^N)^{-2}  H_{N-1}.
$$
It is easy to see that in these coordinates the metric tensor of $H_N$ has diagonal Ricci tensor, thus the $\widehat{H}_i$ do not need quantum corrections. We leave for a further paper a more detailed discussion of the general case.
\end{rmk} 

\section{Conclusions}
By Theorems \ref{T1}, \ref{propB1} and \ref{SepQC} we characterize those natural Hamiltonian systems with quadratic in the momenta constants of motion that are preserved as dynamical constants  during the Laplace-Beltrami quantization, up to the addition of scalar potentials. The examples illustrate both the generality and the usefulness of our approach, not restricted to any a priori assumption of the Hamiltonian quantum correction. In example 3, we consider  a family of  Hamiltonians, arising from the procedure of "extension" of lower-dimensional natural Hamiltonians, admitting a quadratic in the momenta  constant of motion generated by the extension procedure. In a future article we will study the LB quantization of generic extended Hamiltonians with a similar quadratic constant of motion. Some of the results of this study are already introduced here in example 3. We remark that any extended Hamiltonian depends on some freely chosen rational number $k$ and  admits  a constant of motion, polynomial in the momenta, of degree determined by $k$. Therefore, extended Hamiltonians represent an excellent source of examples for any quantization theory, not restricted only to polynomials quadratic in the momenta.  see for example \cite{R}.

\section*{Acknowledgements}

The authors are grateful to Jonathan Kress for suggesting to apply modified quantization to example 2 and to Francisco Herranz for introducing us to the theory of modified LB quantization.



\begin{thebibliography}{12}



\bibitem{He} A. Ballesteros, A. Enciso, F. J. Herranz, O. Ragnisco,  D. Riglioni,  Quantum mechanics on spaces of nonconstant curvature: The oscillator problem and superintegrability, { Ann. Phys.} { 326} n.8, (2011) 2053-2073 

\bibitem{BCRI} S. Benenti, C. Chanu, G. Rastelli,   Remarks on the connection between the additive separation of the Hamilton-Jacobi equation and the multiplicative separation of the Schr\"odinger equation. I. First integrals and symmetry operators, { J. Math. Phys.} { 43} (2002) 5183-5222 


\bibitem{BCRII} S. Benenti, C. Chanu, G. Rastelli,   Remarks on the connection between the additive separation of the Hamilton-Jacobi equation and the multiplicative separation of the Schr\"odinger equation. II. First integrals and symmetry operators, { J. Math. Phys.} { 43} (2002) 5223-5253 

\bibitem{BCRHJ0} S. Benenti, C. Chanu, G. Rastelli, 
Variable-separation theory for the null
Hamilton-Jacobi equation,
{ J. Math. Phys.} { 46} (2005) 042901 

\bibitem{Bl} M. B\l{}aszak, K. Marciniak, Z. Domanski, Separable quantizations of St\"ackel systems, Ann. Phys. { 371} (2016) 460-477

\bibitem{Bl1} M. B\l{}aszak,  K. Marciniak, Classical and quantum superintegrabilty of St\"ackel systems, arXiv:1608.04546 (2016)


\bibitem{CDRfi}  C. Chanu, L. Degiovanni, G. Rastelli,  First integrals of extended Hamiltonians in $(n+1)$-dimensions generated by powers of an operator,  { SIGMA} { 7}  038 (2011) 12 pages 

\bibitem{CDR} C. Chanu, L. Degiovanni, G. Rastelli,  Generalizations of a method for constructing first integrals  of a class of natural Hamiltonians and some remarks about quantization. Journal of Physics: Conference Series, Volume 343,  arxiv:1111.0030 (2011)


\bibitem{CDRfrac} C. Chanu, L. Degiovanni, G. Rastelli, { Extensions of Hamiltonian systems dependent on a rational parameter}, { J. Math. Phys.} { 55} (2014) 122703 

\bibitem{CDRttw} C. Chanu, L. Degiovanni, G. Rastelli, { The Tremblay-Turbiner-Winternitz system as extended Hamiltonian}, { J. Math. Phys.} { 55} (2014) 122701 


\bibitem{CDR1} C. Chanu, L. Degiovanni, G. Rastelli,  Superintegrable extensions of superintegrable systems,  { SIGMA} { 8} , 070 (2012) 12 pages 


\bibitem{Car} B. Carter, Killing tensor quantum numbers and conserved currents in curved space, Phys. Rev. D 16 (1977) 3395 

\bibitem{D1} C. Duval, P. Lecomte, V. Ovsienko,  Conformally invariant quantization. Ann. Inst. Fourier { 49} (1999) 1999-2029 


\bibitem{D2}  C. Duval, V. Ovsienko, Conformally equivariant quantum Hamiltonians. Selecta math. { 7} (2001) 291-320 

\bibitem{D2b} C. Duval, G. Valent, Quantum integrability of quadratic Killing tensors,  J. Math. Phys. { 46} (2005)

\bibitem{D3} C. Duval, G. Valent, A new integrable system on the sphere and conformally equivariant quantization,  J. Geom. Phys. { 61}, no. 8 (2011) 1329-1347  

\bibitem{Ei} L. P. Eisenhart, { Riemannian Geometry}, Princeton U. P. eight edition (1997)

\bibitem{KKMncf} E. G. Kalnins, J. M. Kress, W. Miller Jr, Superintegrability in a non-conformally-flat space, Journal of Physics A: Mathematical and Theoretical, { 46}, 2 (2012)

\bibitem{Rad} J-P. Michel, F. Radoux, J. \v{S}ilhan,
Second Order Symmetries of the Conformal Laplacian,
SIGMA 10, 016 (2014) 26 pages 

\bibitem{R} G. Rastelli, Born-Jordan and Weyl quantizations of the 2D anisotropic harmonic oscillator, SIGMA 12, 081 (2016) 7 pages 

%
%

\end{thebibliography}
\end{document}